\documentclass[letterpaper, 10pt, conference]{ieeeconf}      

\IEEEoverridecommandlockouts                              
\usepackage{amsmath,graphicx,amsfonts,amssymb,epsfig,subfigure,mathrsfs,mathtools}
\usepackage{arydshln}
\usepackage{blkarray}

\usepackage{color}
\usepackage{multirow}
\usepackage{multicol}
\usepackage{lipsum}
\usepackage{rotating}
\usepackage{graphicx}%
\usepackage{algorithm}
\usepackage{algpseudocode}
\usepackage{cite}
\usepackage{tikz}

\newtheorem{theorem}{Theorem}[section]

\newtheorem{definition}{Definition}[section]
\newtheorem{lemma}{Lemma}[section]
\newtheorem{proposition}{Proposition}[section]
\newtheorem{remark}{Remark}[section]

\title{\LARGE \bf
Convergence Analysis of Asynchronous Consensus in Discrete-time Multi-agent Systems with Fixed Topology
}

\author{Kooktae Lee and Raktim Bhattacharya
\thanks{Kooktae Lee and Raktim Bhattacharya are with the Department of Aerospace Engineering, Texas A\&M University,
        College Station, TX 77843-3141, USA, {\tt\scriptsize \{animodor,raktim\}@tamu.edu.}}%
}

\begin{document}
\maketitle
\thispagestyle{empty}
\pagestyle{empty}

\begin{abstract}
In this paper, we study a convergence condition for asynchronous consensus problems in multi-agent systems. 
The convergence in this context implies the asynchronous consensus value converges to the synchronous one and thus is unique. Although it is reported in the literature that the consensus value under asynchronous communications may not coincide with the synchronous consensus value, it has not received much attention. In some applications, the discrepancy between them may result in serious consequences. For such applications it is critical to determine under what conditions the asynchronous consensus value is the same as the synchronous consensus value. We illustrate these issues with a few examples and then provide a condition, which guarantees that the asynchronous consensus value converges to the synchronous one. The validity of the proposed result is verified with simulations.
\end{abstract}

\section{Introduction}
Recently, consensus problems for multi-agent systems have captured attentions from many researchers due to wide range of  applications. In general, a consensus problem is to seek state values that are identical across multiple agents, achieved via communication between agents. The consensus of multi-agent systems, for example, can be adopted to formation control \cite{porfiri2007tracking}, \cite{listmann2009consensus}, flocking control \cite{blondel2005convergence}, \cite{yu2010distributed}, distributed averaging \cite{xiao2004fast}, \cite{nedic2009distributed}, \cite{mehyar2005distributed}, \cite{olshevsky2009convergence}, cooperative control \cite{ren2007consensus}, \cite{ren2008distributed}, synchronization of coupled oscillators \cite{papachristodoulou2010effects}, distributed sensor fusion in the sensor network \cite{xiao2005scheme}, \cite{olfati2005consensus}, \cite{kar2009distributed}, and consensus optimization \cite{zhang2014asynchronous}.

Throughout many research works, it is known that for a fixed topology, the consensus is reachable if and only if the directed graph has a spanning tree \cite{ren2004consensus}. 
This condition ensures that a multi-agent system can reach a consensus by exchanging information between agents among which a connection link exists. One of the big assumptions placed in here is that state updates across the network are synchronized at every time steps, which naturally includes two conditions -- no communication uncertainties in the network such as packet drops or communication delays;  the existence of a global clock that synchronizes time for all agents.
From a practical point of view, however, these conditions are too restrictive or hard to implement. 

Therefore, an asynchronous model \cite{bertsekas1989parallel}, \cite{szyld1998mystery}, \cite{mehyar2005distributed}, \cite{fang2005information}, \cite{gao2013asynchronous}, \cite{zhang2014asynchronous} has arisen as an alternative approach. In this framework, the clock synchronization over multiple agents is unnecessary and communication uncertainties can be also taken into account. 
While many researchers have focused on developing certain conditions under which the asynchronous consensus is reachable, it has been reported in the literature \cite{fang2005information}, \cite{xiao2008asynchronous} that the asynchronous consensus value may not be the same as the synchronous one. These results confirm that the consensus value with asynchronous communication, depends not only on the interaction topology, the initial condition, but also on the randomness in the communication. This leads to different consensus values for every simulation, even when the first two factors are exactly the same. This discrepancy between synchronous and asynchronous consensus values may not be a critical issue in some applications as long as the state of all agents reaches a certain consensus. However, for some other applications such as distributed averaging, parallel fixed-point iteration, consensus optimization, or distributed sensor fusion problems, the asynchronous consensus may converge to an incorrect value and may cause serious consequences. For such systems, it is critical to determine conditions under which consensus values between the synchronous and asynchronous models are identical.

The major contribution of this paper is thus on establishing such conditions to guarantee that the asynchronous consensus value converges into the synchronous case. We apply the switched system framework, the graph theory, and the nonnegative matrix theory to arrive at this condition. Using these tools, we show that the asynchronous consensus value coincides with the synchronous value \textit{irrespective of the randomness in the asynchronous communication}, if there exists \textit{at least} one leader in a multi-agent system. We support the validity of the proposed results with simulations.

Rest of this paper is organized as follows. In section II, we briefly introduce preliminaries with problem descriptions. Section III presents the main results of this paper, which is the convergence condition for asynchronous consensus problems. 
In section IV, we conclude the paper with future works.

\section{Preliminaries and Problem Descriptions}
\noindent\textbf{Notation:} The set of real and natural numbers are denoted by $\mathbb{R}$ and $\mathbb{N}$. Moreover, $\mathbb{N}_0 := \mathbb{N}\cup\{0\}$. The symbols $\text{\underline{I}}^{n\times n}$ and $\text{\underline{0}}^{n\times n}$ stand for the identity and the zero matrix with ${n\times n}$ dimension. Furthermore, the symbol $\rho(\cdot)$ represents the spectral radius of the square matrix.

\subsection{Graph Theory \& Consensus Protocol}
Let $\mathcal{G}=\{\mathcal{V},\mathcal{E},\mathcal{A}\}$ be a directed graph, where $\mathcal{V}=\{\textit{v}_1,\textit{v}_2,\hdots,\textit{v}_n\}$ is the set of $n$ numbers of nodes, $\mathcal{E}\subseteq \mathcal{V}\times \mathcal{V}$ is the set of edges, and $\mathcal{A}\in\mathbb{R}^{n\times n}$ is a weighted adjacency matrix with nonnegative elements $a_{ij}$. A directed edge $e_{ij}\in \mathcal{E}$ denotes the information flow from $\textit{v}_i$ to $\textit{v}_j$.

In the graph theory, a spanning tree is a tree that connects all nodes with edges. A directed graph that has a spanning tree guarantees the reachability of the consensus under the synchronism. We extend the notion of a spanning tree by introducing an $m$-rooted spanning tree, which plays an important role in analyzing the convergence of the asynchronous consensus value.
\begin{definition}[An $m$-rooted spanning tree]
A tree is said to be an $m$-rooted spanning tree if $m$ numbers of roots exist in the tree and all nodes are connected with edges.
\end{definition}
Illustrative examples for an $m$-rooted (directed) spanning tree are given in Fig. \ref{fig: m rooted spanning tree}.
In general, it is known that a (right) stochastic matrix associated with the graph that
has a spanning tree possesses a unique spectral radius unity. Concordantly, a stochastic matrix for an m-rooted spanning
tree contains m numbers of spectral radius unity.

\begin{figure}
\centering
\subfigure[$2$-roots spanning tree]{
\includegraphics[scale=.28]{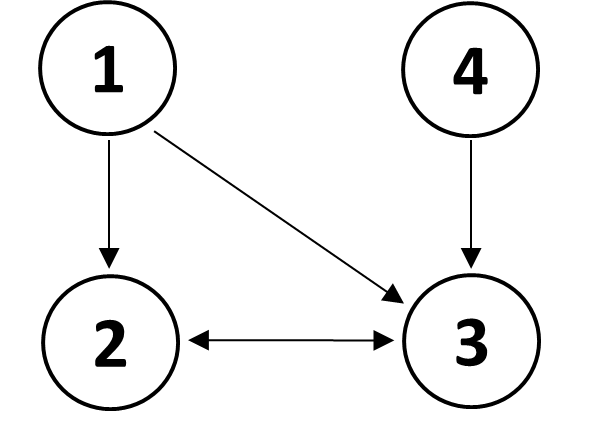}
}\qquad\quad
\subfigure[$3$-roots spanning tree]{
\includegraphics[scale=.28]{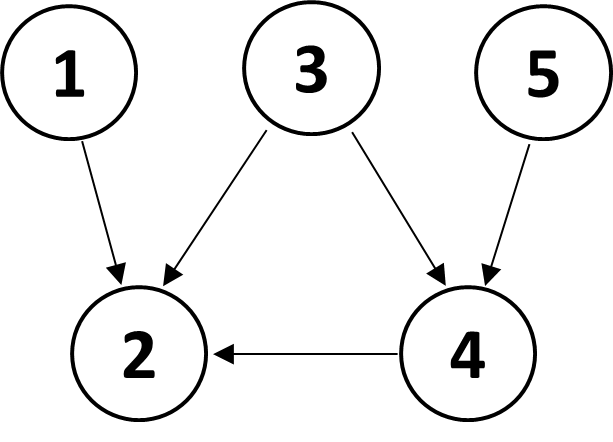}
}
\caption{Examples of $m$-rooted (directed) spanning tree. The roots are given by (a) the node $1$ and $4$; (b) the node $1$, $3$, and $5$}\label{fig: m rooted spanning tree}
\end{figure}

According to \cite{ren2004consensus}, \cite{fang2005information}, for the given interaction topology the discrete-time consensus protocol is represented by
\begin{align}
x_i(k+1) = \dfrac{1}{\sum_{j=1}^{n}a_{ij}(k)}\sum_{j=1}^{n}a_{ij}(k)x_j(k),\label{eqn:consensus protocol}
\end{align}
where $k\in\mathbb{N}_0$ is a discrete-time index, $x_i$ is the state for the $i^{\text{th}}$ agent, and $n$ is a total number of agents. 

Since we consider a directed graph, it may lead to $a_{ij}(k)\neq a_{ji}(k)$, $\forall i, j$.
The network topology is not all-to-all connection and thus, under this interaction topology, information exchange between agents is necessary to reach the consensus $x_1(\bar{k}) = x_2(\bar{k}) = \hdots = x_n(\bar{k})$ for some time $\bar{k}$.

\subsection{Synchronous Consensus Model}

In the absence of communication uncertainties for the networks in which a global clock exists, all state values are synchronized at each discrete time. Then, from \eqref{eqn:consensus protocol} the synchronous consensus protocol can be written by the following compact form.

 $\bullet$ Synchronous model:
\begin{align}
x_{\text{sync.}}(k+1) = F(k)x_{\text{sync.}}(k),\label{eqn:sync}
\end{align}
where $x_{\text{sync.}}=[x_{\text{sync.},1},x_{\text{sync.},2},\hdots, x_{\text{sync.},n}]^{T}\in\mathbb{R}^{n}$, $x_{\text{sync.},i}\in\mathbb{R}$ denotes the state for $i^{\text{th}}$ agent in the synchronous case. The structure of the matrix $F(k)\in\mathbb{R}^{n\times n}$ is given by
\vspace{-0.05in}

\small{
\begin{align*}
F(k):= \begin{bmatrix}
f_{11}(k) & f_{12}(k) & \hdots & f_{1n}(k)\\
f_{21}(k) & f_{22}(k) & \hdots & f_{2n}(k)\\
\vdots & \vdots & \ddots & \vdots\\
f_{n1}(k) & f_{n2}(k) & \hdots & f_{nn}(k)\\
\end{bmatrix}
\end{align*}
}\normalsize
with $f_{ij}(k):=\frac{a_{ij}(k)}{\sum_{j=1}^{n}a_{ij}(k)}$, and each row sum of $F(k)$ is unity. This particular type of the matrix $F(k)$ is also known as the \textit{(right) stochastic matrix (or Markov matrix)}. 

In the synchronous case, the consensus in a multi-agent system is reachable if and only if a directed graph has a spanning tree. Under this condition, the stochastic matrix $F$ has the stationary form, denoted by $F^{\star}:=\lim_{k\rightarrow \infty}F^k = \mathbf{1}\mu^{T}$, where $\mathbf{1}$ is a column vector with all elements are $1$ and $\mu:=[\mu_1,\mu_2,\hdots,\mu_{n}]^{T}$ with the sum of all elements being $1$.

In practice, however, the synchronous model is not appropriate to be implemented because of communication uncertainties. Moreover, a global clock that synchronizes time for all agents may not exist, or may be hard to achieve.  For these reasons, many researchers have investigated an asynchronous model (in \cite{mehyar2005distributed}, \cite{szyld1998mystery,fang2005information,gao2013asynchronous,xiao2008asynchronous}, \cite{zhang2014asynchronous}) that is more desirable and practical than the synchronous one. 

\subsection{Asynchronous Consensus Model}
The dynamics of the asynchronous model in multi-agent systems can be written as follows.

 $\bullet$ Asynchronous model:
\begin{align}
x_{\text{async.},i}(k+1) = f_{ii}&(k)x_{\text{async.},i}(k)\nonumber\\
& + \sum_{j\in\mathcal{N}_i}f_{ij}(k)x_{\text{async.},j}(k_j^*),\label{eqn:each state async}
\end{align}
where $x_{\text{async.},i}\in\mathbb{R}$ denotes the state for the $i^{\text{th}}$ agent in the asynchronous scheme and $\mathcal{N}_i$ is the set of neighbors having connection links with the agent $i$. The sources of asynchrony are represented by the time-delay term $k_j^*\in\mathbb{N}_0$, which is finite and bounded by $\tau_d$, i.e., $k_j^*\in\{k,k-1,\hdots,k-\tau_d\}$. 

In \eqref{eqn:each state async}, the discrete-time index $k$ can be different from each agent. Thus, the global clock is unnecessary. We assume that no asynchrony takes place from the agent itself due to the fact that no communication is required for its own value. Hence, sources for asynchony only come from communications with other agents, which is described by $k_j^*$. In this case, $k_j^*$ becomes a random variable that may affect the consensus value.  

This paper only considers the case where the topology is fixed, i.e., $F(k)=F$, $\forall k\in\mathbb{N}_0$.
Throughout broad research works, the reachability of the consensus under the asynchronism with fixed topology is widely investigated. For example, in \cite{fang2005information} it is shown that the asynchronous consensus is attainable if 
\begin{align}
\rho\left( \,|F - F^{\star}|\,\right) < 1,\label{eqn:stability of async.}
\end{align}
where the symbol $|\cdot|$ denotes element-wise absolute values for the given matrix.

Notice that \eqref{eqn:stability of async.} is a sufficient condition for the reachability of the asynchronous consensus and does not guarantee a unique consensus under asynchronous communications. As pointed out in \cite{fang2005information}, the necessity for the uniqueness of the consensus value is given by $\rho(F) = 1$. Thus, the asynchronous consensus value could be different from the synchronous one. We illustrate this issue with the following example.

\noindent Example $1$: Consider a consensus problem with a fixed interaction topology, where an associated matrix $F$ is given by
\small{
\begin{align*}
F = \begin{bmatrix}
0.5 & 0.5 & 0 & 0 & 0\\
0.4 & 0.3 & 0 & 0 & 0.3\\
0.1 & 0.2 & 0.2 & 0.4 & 0.1\\
0 & 0 & 0 & 0.7 & 0.3\\
0.1 & 0.5 & 0.1 & 0.2 & 0.1
\end{bmatrix}.
\end{align*}
}\normalsize
An initial state condition is set as $x_{\text{sync.}}(0) = x_{\text{async.}}(0) = [3, 2, 1, 3, 5]^{T}$.


The consensus is reachable for both synchronous and asynchronous cases, since the interaction topology has a spanning tree and $\rho\big( |F - \mathbf{1}\mu^{T}|\big) = 0.83 < 1$, where $\mu=[0.32,\,    0.35,\,    0.02,\,    0.14,\,    0.17]^{T}$. In the case of the asynchronous model, Monte Carlo simulation was carried out with a maximum delay $\tau_d=5$ and a uniform distribution for $k_j^*\in\{k,k-1,\hdots,k-5\}, \, j=1,2,\hdots,5$, to describe the random behavior of asynchrony. In Fig. \ref{fig:2},
the trajectories of $2$-norm value of the state vector is depicted for both synchronous and asynchronous models. As shown in Fig. \ref{fig:2}, asynchronous consensus values are sensitive to the randomness of the delay and do not match the synchronized value. In fact, the asynchronous consensus value is a \textit{random variable} that depends on the randomness of the communication channel.
This example clearly demonstrates that the asynchronous consensus value is determined by initial condition, the interaction topology, and the  randomness of communication, whereas the synchronous consensus relies only on the first two factors. 

This discrepancy is unacceptable in applications such as distributed averaging, parallel fixed-point iteration, distributed sensor fusion, or average consensus in multi-agent formation or flocking control problems. To motivate the problem further, we consider these examples in more detail and highlight why the discrepancy in consensuses is an issue.

\begin{figure}
\centering
\includegraphics[scale=0.45]{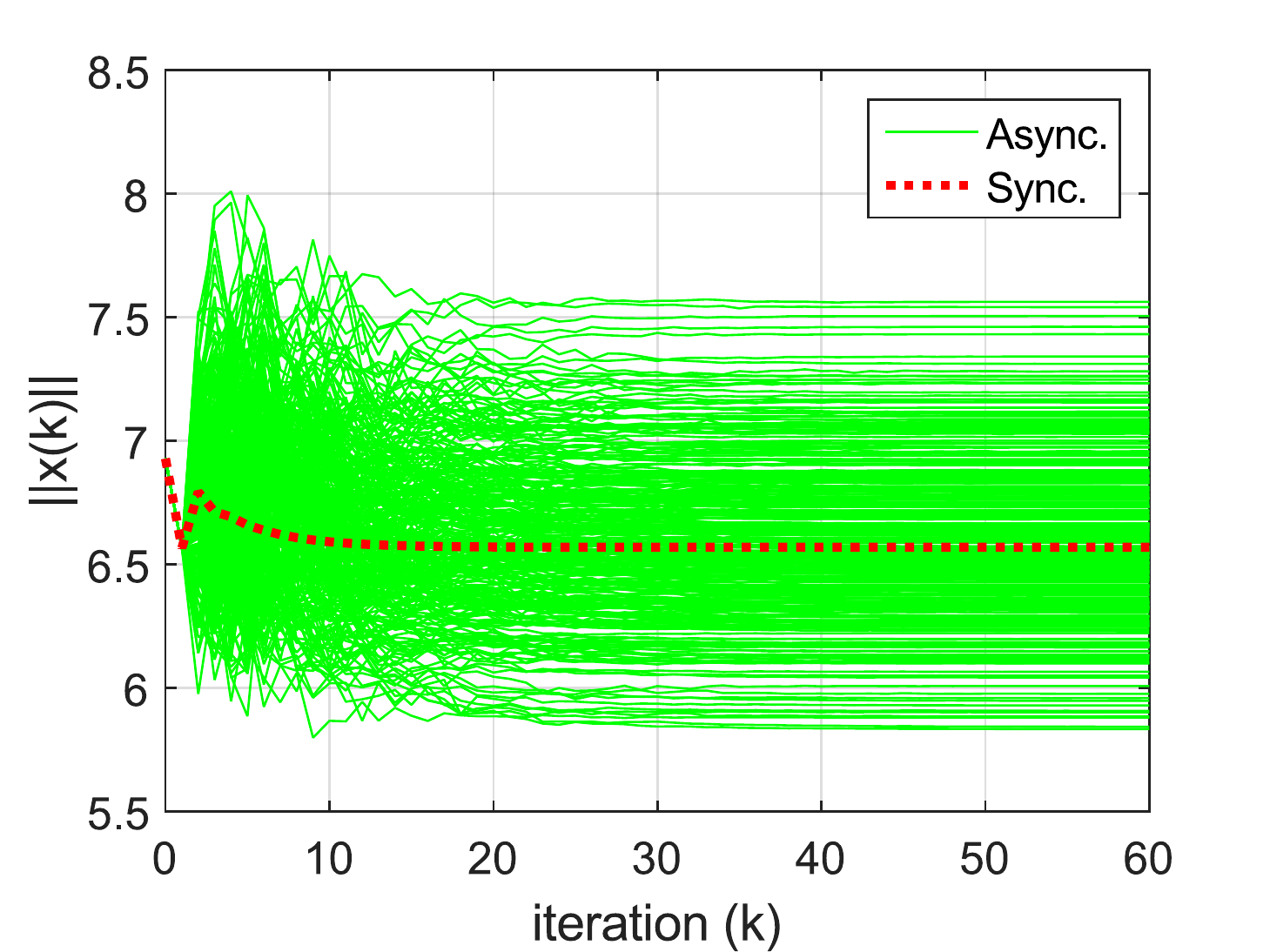}
\caption{The trajectory comparison between the asynchronous model with $300$ samples (solid lines) and the synchronous model (dotted line) for Example $1$}\label{fig:2}
\end{figure}

\subsection{Consensus Critical Applications}
\textbf{1) Distributed averaging over communication networks} -- Suppose that there are $n$ numbers of agents in the network with fixed topology. Initially, each agent has a different value, and the objective is to find out the global average value \cite{xiao2004fast}, \cite{nedic2009distributed}, \cite{mehyar2005distributed}, \cite{olshevsky2009convergence}. In the absence of communication uncertainties such as packet drops or communication delays in the network, the average can be obtained by synchronized information exchange between agents. In practice, however, such communication uncertainties may take place, leading to asynchronous updates. Consequently, this results in the incorrect average.

\textbf{2) Parallel computing for fixed-point iteration} --  
Parallel computing is widely used technique to speedup computation of fixed-point iterations. For example, in \cite{lee2015switched} and \cite{lee2015asynchronous}, one-dimensional heat equation using a finite difference method is solved by parallel computing in which a certain group of grid points for the finite difference scheme is assigned to each CPU or GPU core. Thus, in parallel computing each value for the group of grid points is computed by different cores, followed by communication between cores for updating values at the boundary grid points of the group. Starting from some initial condition, the spatio-temporal evolution of temperature will converge to steady-state temperature (i.e. a consensus value) as iterations progress. In this case, asynchronous updates can improve the computing performance by avoiding synchronization bottleneck problem, but consequently leads to different consensus value (i.e. solution to the heat equation)\cite{lee2015asynchronous}. Thus, we get an erroneous solution to the problem in the asynchronous updates. This issue will also appear in other fixed-point iterations such as numerical method, optimization, signal processing, etc.

\textbf{3) Distributed sensor fusion in the sensor network} -- Distributed sensor fusion is one of the  applications of consensus problems \cite{xiao2005scheme}, \cite{olfati2005consensus}, \cite{kar2009distributed}. For instance, distributed sensors can be used to track a target that is of interest. In this case, each sensor can exchange information about the target through the network with a given interaction topology that describes the weight between distributed sensors. The asynchronous consensus value may not converge into the synchronous one, if the information from distributed sensors is updated asynchronously. Thus, the target information will be erroneous.

%
\textbf{4) Average consensus in multi-agent formation or flocking control} -- Consensus problems have been broadly adopted to formation or flocking control in multi-agent systems \cite{porfiri2007tracking,listmann2009consensus,yu2010distributed,blondel2005convergence}, \cite{ren2004consensus}. In this case, an alignment of heading or a velocity consensus is of concern. On top of the consensus itself, consider the case that an average consensus is required (e.g., average velocity or average heading). Although the consensus is reachable by communications between agents, the asynchronous scheme may deliver an inexact average consensus value as reported in Example $1$.

Therefore, for some applications, we have to guarantee that the asynchronous consensus value is unique and is the same as the synchronized case. This paper primarily focuses on developing such conditions and the main results are introduced in the following section.


\section{Convergence Analysis}
We first convey the concept of the switched system, which is adopted to analyze the asynchronous consensus value.
\subsection{Switched System}
For simplicity, suppose that there exists a two-agent system of which state is given by $x(k) = [x_1(k),\,x_2(k)]^{T}$. The stochastic matrix $F$ associated with a given interaction topology for this system is written as
\begin{align*}
F = \begin{bmatrix}
f_{11} & f_{12}\\
f_{21} & f_{22}
\end{bmatrix}.
\end{align*}

If we consider a concatenated state vector $\tilde{x}(k) := [x(k)^{T}, x(k-1)^{T}, \hdots, x(k-\tau_d)^{T}]^{T}$, then the dynamics of $\tilde{x}$ is updated by the following form:
\begin{align}
\tilde{x}(k+1) &= W_{\sigma_k}\tilde{x}(k),\label{eqn:switched system}
\end{align}
where $W_{\sigma_k}$ denotes a modal matrix with a switching mode $\sigma_{k}$  at time $k$.
When the maximum delay is given by $\tau_d=1$, the modal matrix $W_{\sigma_k}$ at any time $k$ becomes one of the following form

{\small
\begin{align}
&W_1 = \left[ \begin{array}{c:c}
f_{11} \,\, f_{12} & 0 \quad 0\\
f_{21} \,\, f_{22} & 0 \quad 0\\
\hdashline
& \\ 
  \text{\Large{\underline{I}}}^{\normalsize n\times n} & \text{\Large{\underline{0}}}^{n\times n} \\ 
&\\
\end{array} \right],\:
W_2 = \left[ \begin{array}{c:c}
f_{11} \quad 0 & 0 \,\, f_{12}\\
f_{21} \,\, f_{22} & 0 \quad 0\\
\hdashline
& \\ 
  \text{\Large{\underline{I}}}^{\normalsize n\times n} & \text{\Large{\underline{0}}}^{n\times n} \\ 
&\\
\end{array} \right],\nonumber\\
&W_3 = \left[ \begin{array}{c:c}
f_{11} \,\, f_{12} & 0 \quad 0\\
0 \quad f_{22} & f_{21} \, 0\\
\hdashline
& \\ 
  \text{\Large{\underline{I}}}^{\normalsize n\times n} & \text{\Large{\underline{0}}}^{n\times n} \\ 
&\\
\end{array} \right],\:
W_4 = \left[ \begin{array}{c:c}
f_{11} \,\, 0 & 0 \,\, f_{12}\\
0 \,\,\, f_{22} & f_{21} \,\, 0\\
\hdashline
& \\ 
  \text{\Large{\underline{I}}}^{\normalsize n\times n} & \text{\Large{\underline{0}}}^{n\times n} \\ 
&\\
\end{array} \right].\label{eqn:W_j}
\end{align}
}

Notice that above modal matrices $W_j$, $j=1,2,3,4$, are obtained by taking into account every possible scenarios for delays.
Based on the assumption that asynchrony sources only come from communications with connected agents, the diagonal elements $f_{ii}$ for $F$ is always fixed in all $W_j$, whereas off-diagonal elements for $F$ move accordingly in $W_j$. The stochastic property of asynchrony is then represented by a certain switching rule (e.g, Markovian switching) that governs the switching process $\{\sigma_k\}$ for the given  switching mode $\sigma_k\in\{1,2,3,4\}$. The switched system with a stochastic switching process is particularly referred to as the stochastic switched system or stochastic jump linear system \cite{lee2015performance}. 
Although the above example is derived when the number of agents are two with the maximum delay $\tau_d=1$, the most general case ($n$ numbers of agents with the maximum delay $\tau_d\in\mathbb{N}$) is induced in this context
with the following modal matrix structure:

\vspace{0.1in}
\noindent$W_{\sigma_k} \in\mathbb{R}^{n(\tau_d+1)\times n(\tau_d+1)}$
\small{
\begin{align}
&=\begin{bmatrix}
W_{11}(k) & W_{12}(k) & W_{13}(k) & \cdots & W_{1(\tau_d+1)}(k)\\
\text{{\underline{I}}}^{n\times n} & \text{{\underline{0}}}^{n\times n}& \text{{\underline{0}}}^{n\times n} & \cdots &\text{{\underline{0}}}^{n\times n}\\ 
\text{{\underline{0}}}^{n\times n} & \text{{\underline{I}}}^{n\times n}& \text{{\underline{0}}}^{n\times n} & \cdots &\text{{\underline{0}}}^{n\times n}  \\
\vdots & \text{{\underline{0}}}^{n\times n} &\ddots & \ddots& \vdots\\
\text{{\underline{0}}}^{n\times n} & \text{{\underline{0}}}^{n\times n} & \ddots & \text{{\underline{I}}}^{n\times n} & \text{{\underline{0}}}^{n\times n}
\end{bmatrix},&\hfill\label{eqn:W_j general}
\end{align}
}\normalsize
where $W_{1r}(k)\in\mathbb{R}^{n\times n}$ in $W_j$ is a block matrix satisfying $\sum_{r=1}^{\tau_d+1}W_{1r}(k) = F$ and the diagonal elements in $W_{11}(k)$ is identical with that in $F$ for all $k\in\mathbb{N}_0$.

The advantage of implementing the switched system framework is that one can formulate the inherent dynamics of the asynchronous model   
in this framework. Thus, we can analyze the convergence of asynchronous consensus through the switched system with \textit{arbitrary switching process}. Starting from a given initial condition $\tilde{x}(0)$, the dynamics of the switched system \eqref{eqn:switched system} can be also written as
\begin{align}
\tilde{x}(k+1) =  \tilde{W}(k)\tilde{x}(0),\label{eqn:W_tilde}
\end{align}
where $\tilde{W}(k) := W_{\sigma_k}W_{\sigma_{k-1}}\cdots W_{\sigma_{1}}W_{\sigma_{0}}$ and $\{\sigma_r\}_{r=0}^{k}$ denotes the random switching process.
Interestingly, all $W_{\sigma_k}$ matrices still form the stochastic matrix, since each element of $W_{\sigma_k}$ is nonnegative and row sum is always unity as described in \eqref{eqn:W_j}.

\subsection{Nonnegative Matrix Property}
For the given stochastic matrix $W_{\sigma_k}$, the product of $W_{\sigma_k}$ in time, denoted by $\tilde{W}(k)$, also construct the stochastic matrix by the following lemma.
\begin{lemma}\label{lemma:prod of stochastic mat.}
The product of any stochastic matrices also forms the stochastic matrix.
\end{lemma}
\begin{proof}
We consider any stochastic matrices $B\in\mathbb{R}^{n\times n}$ and $C\in\mathbb{R}^{n\times n}$ given as follows:
\vspace{-0.12in}

%
%
\small{
\begin{align}
&B=\begin{bmatrix}
b_{11} & b_{12} & \hdots & b_{1n}\\
b_{21} & b_{22} & \hdots & b_{2n}\\
\vdots &\vdots & \ddots & \vdots \\
b_{n1} & b_{n2} & \hdots & b_{nn}
\end{bmatrix}, \,\,
C=\begin{bmatrix}
c_{11} & c_{12} & \hdots & c_{1n}\\
c_{21} & c_{22} & \hdots & c_{2n}\\
\vdots &\vdots & \ddots & \vdots \\
c_{n1} & c_{n2} & \hdots & c_{nn}
\end{bmatrix},\label{eqn:matrix A and B}
\end{align}
}\normalsize\vspace{-0.1in}
\begin{align*}
\text{where }\qquad
&0 \leq b_{ij} \leq 1 \,\,\text{ and }\,\, 0 \leq c_{ij} \leq 1,\,\, \forall i,j,\\
&\sum_{j=1}^{n}b_{ij} = 1 \text{ and } \sum_{j=1}^{n}c_{ij} = 1,\,\, \forall i.\qquad\qquad
\end{align*}

If we define a new matrix $D := BC$ of which element is given by $d_{ij}$, then it satisfies
\begin{align*}
0 \leq d_{ij} &= \sum_{r=1}^{n}b_{ir}c_{rj}\leq \sum_{r=1}^{n}b_{ir}\cdot 1 = 1.
\end{align*}
Thus, we have $0\leq d_{ij} \leq 1$, $\forall i,j$.

Moreover, the row sum of $D$ is given by
\begin{align*}
\sum_{j=1}^{n}d_{ij} &= \sum_{j=1}^{n}\sum_{r=1}^{n}b_{ir}c_{rj}= \sum_{r=1}^{n}b_{ir}\sum_{j=1}^{n}c_{rj}=1.
\end{align*}
As a consequence, the matrix for the product of two stochastic matrices has nonnegative elements with row sum unity, which is the stochastic matrix.
\end{proof}

According to Lemma \ref{lemma:prod of stochastic mat.}, $\tilde{W}(k)$ becomes the stochastic matrix at any time $k$.
Furthermore, the following proposition is developed to denote the particular structure of $\tilde{W}(k)$.

\begin{proposition}\label{prop:W_tilde_star}
For the matrix $\tilde{W}(k)$, defined in \eqref{eqn:W_tilde}, 
consider the structure of $\tilde{W}(k)\in\mathbb{R}^{n(\tau_d+1)\times n(\tau_d+1)}$ given by
\footnotesize{
\begin{align*}
\tilde{W}(k) = \begin{bmatrix}
\tilde{W}_{11}(k) & \tilde{W}_{12}(k) & \hdots & \tilde{W}_{1(\tau_d+1)}(k)\\
\tilde{W}_{21}(k) & \tilde{W}_{22}(k) & \hdots & \tilde{W}_{2(\tau_d+1)}(k)\\
\vdots & \vdots & \ddots & \vdots\\
\tilde{W}_{(\tau_d+1)1}(k) & \tilde{W}_{(\tau_d+1)2}(k) & \hdots & \tilde{W}_{(\tau_d+1)(\tau_d+1)}(k)
\end{bmatrix},
\end{align*}
}\normalsize
where $\tilde{W}_{ij}(k)\in\mathbb{R}^{n\times n}$ denotes the block matrix in $\tilde{W}(k)$.
Then, $\tilde{W}(k)$ can be written by the first row block matrices in $\tilde{W}(k-1),\tilde{W}(k-2), \hdots, \tilde{W}(k-\tau_d)$ as follows:
\footnotesize{
\begin{align*}
\tilde{W}(k) = \begin{bmatrix}
\tilde{W}_{11}(k) & \tilde{W}_{12}(k) & \hdots & \tilde{W}_{1(\tau_d+1)}(k)\\
\tilde{W}_{11}(k-1) & \tilde{W}_{12}(k-1) & \hdots & \tilde{W}_{1(\tau_d+1)}(k-1)\\
\vdots & \vdots & \ddots & \vdots\\
\tilde{W}_{11}(k-\tau_d) & \tilde{W}_{12}(k-\tau_d) & \hdots & \tilde{W}_{1(\tau_d+1)}(k-\tau_d)
\end{bmatrix}.
\end{align*}
}\normalsize
\end{proposition}
\vspace{0.2in}

\begin{proof}
For simplicity, we consider the case when $\tau_d = 1$. The most general case is then obtained by induction. For $\tau_d = 1$, $\tilde{W}(k)$ is given by
\small{
\begin{align*}
\tilde{W}(k) &= W_{\sigma_k}\tilde{W}(k-1)\\
&= \begin{bmatrix}
W_{11}(k) & W_{12}(k)\\
\\
\text{\large{\underline{I}}}^{n\times n} & \text{\large{\underline{0}}}^{n\times n}
\end{bmatrix}
\begin{bmatrix}
\tilde{W}_{11}(k-1) & \tilde{W}_{12}(k-1)\\
\\
\tilde{W}_{21}(k-1) & \tilde{W}_{22}(k-1)\\
\end{bmatrix}\\
&= \begin{bmatrix}
\tilde{W}_{11}(k) & \tilde{W}_{12}(k)\\
\\
\tilde{W}_{11}(k-1) & \tilde{W}_{12}(k-1)\\
\end{bmatrix},
\end{align*} 
}\normalsize
where $W_{ij}(k)\in\mathbb{R}^{n\times n}$ denotes the $i^{\text{th}}$ row and $j^{\text{th}}$ column block matrix in $W_{\sigma_k}$, and $\tilde{W}_{11}(k) = \sum_{r=1}^{2}W_{1r}(k)\tilde{W}_{r1}(k-1)$, $\tilde{W}_{12}(k) = \sum_{r=1}^{2}W_{1r}(k)\tilde{W}_{r2}(k-1)$.
\end{proof}

If the asynchronous scheme is guaranteed to be stable by the condition \eqref{eqn:stability of async.}, $\tilde{W}(k)$ arrives at the stationary form, defined by $\tilde{W}^{\star}:=\lim_{k\rightarrow\infty}\tilde{W}(k)$, which can be written from Proposition \ref{prop:W_tilde_star} as follows:

{\small{
\begin{align}
\tilde{W}^{\star} =
\begin{bmatrix}
\tilde{W}_{11}^{\star} & \tilde{W}_{12}^{\star} & \hdots & \tilde{W}_{1(\tau_d+1)}^{\star}\\
\tilde{W}_{11}^{\star} & \tilde{W}_{12}^{\star} & \hdots & \tilde{W}_{1(\tau_d+1)}^{\star}\\
\vdots & \vdots & \ddots & \vdots\\
\tilde{W}_{11}^{\star} & \tilde{W}_{12}^{\star} & \hdots & \tilde{W}_{1(\tau_d+1)}^{\star}
\end{bmatrix}.\label{eqn:invariant form of Wtilde}
\end{align}
}}\normalsize

Once $\tilde{W}(k)$ reaches the invariant measure $\tilde{W}^{\star}$, the equality $\tilde{W}^{\star} = W_{\sigma_k}\tilde{W}^{\star}$ holds for any $\sigma_k$. In other words, $\tilde{W}^{\star}$ stays invariant with respect to the left multiplication of $W_{\sigma_k}$, regardless of $\sigma_k$.

\subsection{Convergence Condition}
We present the main result that provides some conditions under which the asynchronous consensus value converges into the synchronous one.
\begin{theorem}\label{theorem:async. consensus with multi-leader}
Consider the dynamics of the asynchronous model \eqref{eqn:each state async} for a multi-agent system with fixed topology. 
For the stable asynchronous scheme satisfying \eqref{eqn:stability of async.}, the asynchronous state value converges into the synchronous one \textit{regardless of asynchrony}, if the interaction topology has $m$-rooted spanning tree.
\end{theorem}
\begin{proof}
For a given $m$-rooted spanning tree, the interaction topology can be reconstructed by reordering the agents' number such that $f_{ii} = 1$ and $f_{ij}=0$ for $i=1,2,\hdots, m$, $\forall j\in\{1,2,\hdots,n\}$ as follows:
\begin{align}
F = \left[ \begin{array}{cl}
\quad\text{\large\underline{I}}^{m\times m}& \text{\large\underline{0}}^{m\times (n-m)}\\
X & Y
\end{array} \right]\in\mathbb{R}^{n\times n},\label{eqn:F for m-rooted spanning tree}
\end{align}
where $X\in\mathbb{R}^{(n-m)\times m}$ and $Y\in\mathbb{R}^{(n-m)\times (n-m)}$ are nonnegative matrices satisfying each row sum of $[X\,\, Y]$ to be unity. Note that the submatrix $X$ cannot be a zero matrix, since the topology map is a spanning tree (i.e., all nodes are connected by network edge). 

For such a matrix $F$, the stationary form $F^{\star}:=\lim_{k\rightarrow\infty}F^k$ is obtained by
\begin{align*}
F^{\star}= \left[ \begin{array}{cl}
\quad\text{\large\underline{I}}^{m\times m}& \quad\text{\large\underline{0}}^{m\times (n-m)}\\
\displaystyle\sum_{k=0}^{\infty}Y^kX & \quad\displaystyle\lim_{k\rightarrow\infty}Y^k
\end{array} \right].
\end{align*}

The characteristic polynomial of $F$ is then written as $P(\lambda) = \text{det}(\lambda\, \text{\underline{I}}^{n\times n} - F)$. By the determinant property of the block matrix that $\text{det}\left(\scriptsize{\begin{matrix}A&0\\C&D\end{matrix}}\right) = \text{det}\left(A\right)\cdot\text{det}\left(D\right)$, it follows
\begin{align*}
P(\lambda) &= \text{det}\big(\lambda\, \text{\underline{I}}^{m\times m} - \text{\underline{I}}^{m\times m}\big)
\cdot\text{det}\big(\lambda\, \text{\underline{I}}^{(n-m)\times (n-m)} - Y \big)\\
& = (\lambda-1)^{m}\cdot \text{det}\big(\lambda\, \text{\underline{I}}^{(n-m)\times (n-m)} - Y \big).
\end{align*}
The eigenvalue of $F$ then satisfies $P(\lambda) = 0$, resulting in $m$ numbers of eigenvalue unity from $(\lambda-1)^{m}=0$.
Since $F$ is an $m$-rooted spanning tree that has $m$ numbers of spectral radius unity, this leads to $\text{det}\big(\text{\underline{I}}^{(n-m)\times (n-m)} - Y \big) \neq  0$. Therefore, it is guaranteed that $\rho(Y) < 1$ by the following result:

Suppose $\lambda_i$ and $v_i$, $i=1,2,\hdots, (n-m)$, respectively, is the eigenvalue and eigenvector for the matrix $\big(\text{\underline{I}}^{(n-m)\times (n-m)}-Y\big)$. Then, by the eigenvalue-eigenvector relationship, we have
\begin{align*}
\big(\text{\underline{I}}^{(n-m)\times (n-m)}-Y\big)v_i = \lambda_iv_i \Longleftrightarrow  \,Yv_i = \tilde{\lambda}_i v_i,
\end{align*}
where $\tilde{\lambda}_i := 1-\lambda_i$. Hence, it is clear that $\tilde{\lambda}_i$ is an eigenvalue for $Y$ with a corresponding eigenvector $v_i$.  Due to the fact that $\text{det}\big(\text{\underline{I}}^{(n-m)\times (n-m)} - Y \big)\neq 0$, we have $\lambda_i = 1 - \tilde{\lambda}_i \neq 0$, $\forall i$. Thus, $\tilde{\lambda}_i \neq 1$. Finally, the Perron-Frobenius theorem with the nonnegativeness of $Y$ yields the spectral radius of $Y$ being strictly less than unity, i.e., $\rho(Y) < 1$.

When $\rho(Y)<1$, it is known that $\sum_{k=0}^{\infty}Y^kX = \big(\text{\underline{I}}^{(n-m)\times (n-m)}-Y\big)^{-1}X$, which further leads to
\begin{align*}
F^{\star} = \left[ \begin{array}{cl}
\quad\text{\large\underline{I}}^{m\times m}& \quad\text{\large\underline{0}}^{m\times (n-m)}\\
(\text{\underline{I}}^{(n-m)\times (n-m)}-Y)^{-1}X & \quad \text{\large\underline{0}}^{(n-m)\times (n-m)}
\end{array} \right].
\end{align*}

Now, we consider the invariant measure of $\tilde{W}(k)\in\mathbb{R}^{n(\tau_d+1)\times n(\tau_d+1)}$. For the synchronous case, $\tilde{W}_{\text{sync.}}(k)$ is obtained by $\tilde{W}_{\text{sync.}}(k) = W_1^k$, where $W_{1}$ is the modal matrix given in \eqref{eqn:W_j}, which corresponds to the case that $\tau_d=0$ (i.e., no asynchrony). Then, the stationary form $\tilde{W}_{\text{sync.}}^{\star}$ is calculated by
\begin{align}
\tilde{W}_{\text{sync.}}^{\star} =  \lim_{k\rightarrow\infty}W_1^k= \begin{bmatrix}
F^{\star} & $\underline{0}$^{\,n\times (n-1)(\tau_d+1)}\\
\vdots & \vdots\\
F^{\star} & $\underline{0}$^{\,n\times (n-1)(\tau_d+1)}
\end{bmatrix}.\label{eqn:W_tilde_sync_star}
\end{align}

On the other hands, in the asynchronous case the equation $\tilde{W}^{\star} = W_{\sigma_k}\tilde{W}^{\star}$ and \eqref{eqn:invariant form of Wtilde} leads to
\begin{align}
\tilde{W}_{11}^{\star} &= \sum_{r=1}^{\tau_d+1}W_{1r}(k)\tilde{W}_{r1}^{\star}= \left(\sum_{r=1}^{\tau_d+1}W_{1r}(k)\right)\tilde{W}_{11}^{\star}= F\tilde{W}_{11}^{\star}\label{eqn:W_tilde_11_0}
\end{align}

By the given form of $F$, the structure of $\tilde{W}_{11}^{\star}$ becomes
\begin{align}
\tilde{W}_{11}^{\star}:= \left[ \begin{array}{cl}
\quad\text{\large\underline{I}}^{m\times m}& \quad\text{\large\underline{0}}^{m\times (n-m)}\\
R^{\star} & \quad S^{\star}
\end{array} \right]\label{eqn:W_tilde_11}
\end{align}
with $R^{\star}\in\mathbb{R}^{(n-m)\times m}$ and $S^{\star}\in\mathbb{R}^{(n-m)\times (n-m)}$.

Thus, from \eqref{eqn:F for m-rooted spanning tree}, \eqref{eqn:W_tilde_11_0}, and \eqref{eqn:W_tilde_11}, the matrix $R^{\star}$ and $S^{\star}$, respectively, satisfies
\begin{align}
&R^{\star} = X + YR^{\star} \Longleftrightarrow \big(\text{\underline{I}}^{(n-m)\times (n-m)}-Y\big)R^{\star} = X\label{eqn:(I-Y)R}
\end{align}
and
\begin{align}
&S^{\star} = YS^{\star} \Longleftrightarrow \big(\text{\underline{I}}^{(n-m)\times (n-m)}-Y\big)S^{\star} = \text{\underline{0}}^{(n-m)\times (n-m)}.\label{eqn:(I-Y)S}
\end{align}

Since $\text{det}\big(\text{\underline{I}}^{(n-m)\times (n-m)} - Y \big)\neq 0$, the matrix $\big(\text{\underline{I}}^{(n-m)\times (n-m)}-Y\big)$ in \eqref{eqn:(I-Y)R} and \eqref{eqn:(I-Y)S} is invertible and hence, 
$\tilde{W}_{11}^{\star} = F^{\star}$. 
Knowing that $\tilde{W}_{\text{sync.}}^{\star}$ is also the stochastic matrix, it is guaranteed that each row sum of $F^{\star}$ is unity from \eqref{eqn:W_tilde_sync_star}, which leads to $\tilde{W}_{1j}^{\star} = \text{\underline{0}}^{n \times n}$, for $j=2,3,\hdots, n$. 
As a result, we conclude that $\tilde{W}^{\star} = \tilde{W}_{\text{sync.}}^{\star}$.
\end{proof}

\begin{remark}
In \cite{fang2005information}, it is claimed without proof that the interaction topology having a rooted spanning tree ensures the convergence of the asynchronous consensus. This condition corresponds to a multi-agent system with a single leader case.
Theorem \ref{theorem:async. consensus with multi-leader} extends this condition to a more general one that is the existence of \textit{at least} one leader in multi-agent systems. Thus, Theorem \ref{theorem:async. consensus with multi-leader} guarantees that the asynchronous consensus value converges into the synchronous consensus value, \textit{irrespective of asynchrony}, for single-leader as well as multi-leader systems.
\end{remark}
\begin{remark}
In the multi-leader case, it is not a consensus problem anymore if initial values for multiple leaders are different. The states values of leaders remain time-invariant and hence, they never reach the consensus. Nonetheless, Theorem \ref{theorem:async. consensus with multi-leader} can be still applicable in finding solutions for multi-leader-follower systems associated with coordination problems, cooperative control problems, and even parallel fixed-point iterations, under the asynchronism.
\end{remark}
\vspace{0.05in}

To validate the proposed result, we revisit Example $1$. The only difference here is that there are two leaders in the node $1$ and $4$ as given in the following example.

\noindent Example $2$: Consider a consensus problem with a fixed interaction topology, where the corresponding stochastic matrix is given by
\vspace{-0.1in}

{\small
\begin{align*}
F = \begin{bmatrix}
1 & 0 & 0 & 0 & 0\\
0.4 & 0.3 & 0 & 0 & 0.3\\
0.1 & 0.2 & 0.2 & 0.4 & 0.1\\
0 & 0 & 0 & 1 & 0\\
0.1 & 0.5 & 0.1 & 0.2 & 0.1
\end{bmatrix}.
\end{align*}
}

Although this system has multiple leaders (node $1$ and $4$), it is a consensus problem, since the initial state values for two different leaders are identically given as $x_1(0)=x_4(0)=3$. Monte Carlo simulations are again performed for the realization of randomness in the asynchronous model. As depicted in Fig. \ref{fig:3}, the trajectories of $2$-norm value of the state vector for the asynchronous scheme converge into that for the synchronous one. While transient jitters are observed under asynchronism, all state values finally reach the consensus $x_1 = x_2 = \hdots = x_5 = 3$, which corresponds to the synchronous consensus value as well, regardless of asynchrony.
\begin{figure}
\centering
\includegraphics[scale=0.45]{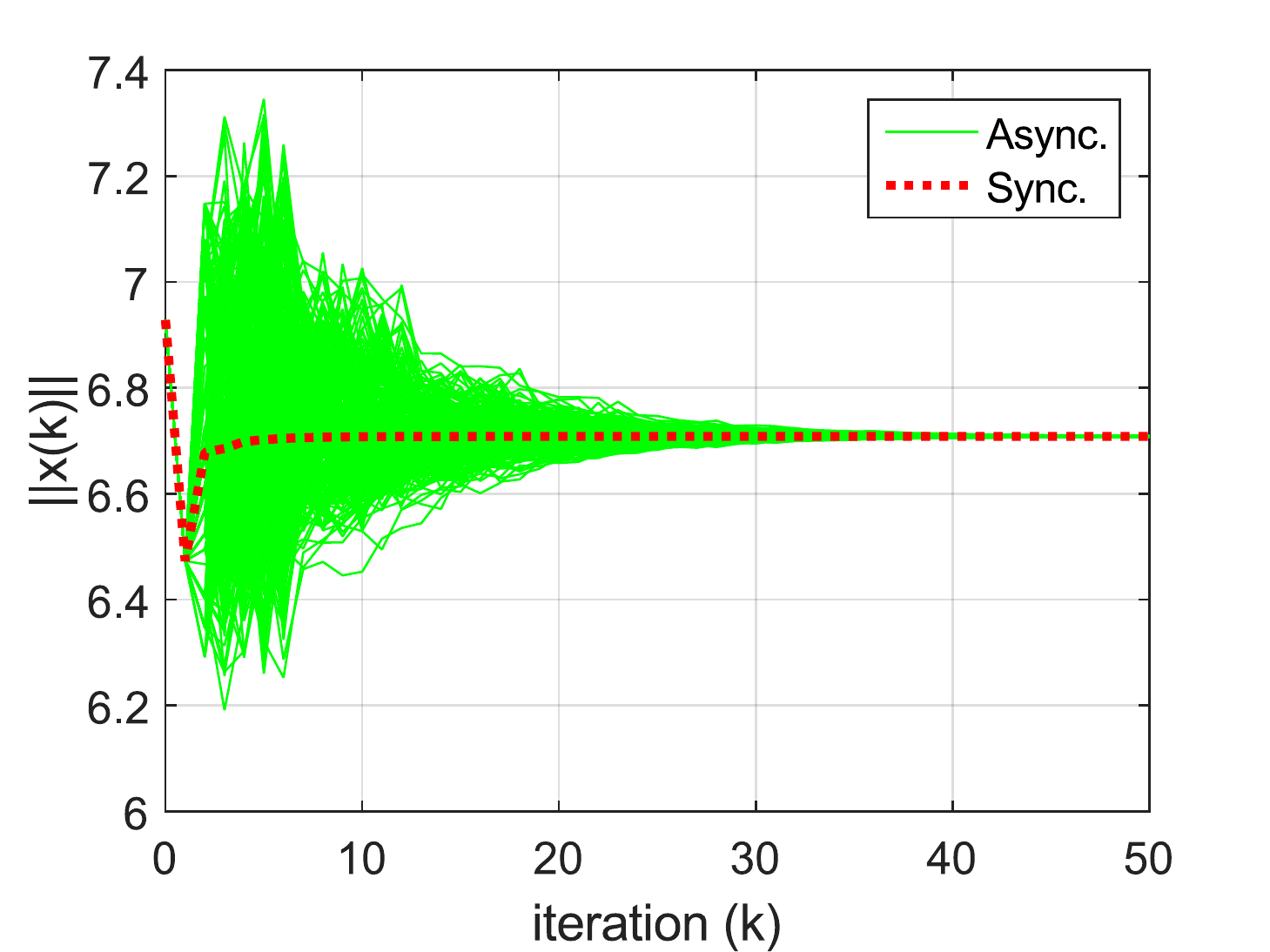}
\caption{The trajectory comparison between the asynchronous model with $300$ samples (solid lines) and the synchronous model (dotted line) for Example $2$}\label{fig:3}
\end{figure}


\section{Concluding Remarks and Future Works}
This paper provides a sufficient condition for the consistency between the synchronous and asynchronous consensus values in multi-agent systems. We proved that the asynchronous consensus value converges to the synchronous value, irrespective of randomness in the asynchronous communication, if \textit{at least} one leader exists in the multi-agent system. We conjecture that this is also a necessary condition for the consistency. In this work, the interaction topology is assumed to be fixed, whereas from the practical perspective, a time-varying topology should be considered, which will be the focus of our future work. 

\bibliographystyle{ieeetr}
\bibliography{ref_CDC2016}

\end{document}